\documentclass[letterpaper, 10 pt, conference]{ieeeconf}
\usepackage{amsmath,amsfonts}
\usepackage{algorithmic}
\usepackage{algorithm}
\usepackage{array}
\usepackage{textcomp}
\usepackage{stfloats}
\usepackage{url}
\usepackage{verbatim}
\usepackage{graphicx}
\usepackage{cite}
\usepackage{tabularx}
\usepackage{mathrsfs}
\usepackage{xcolor}
\usepackage{mathtools}
\usepackage{tikz}
\usetikzlibrary{positioning, shapes.geometric, arrows.meta}

\IEEEoverridecommandlockouts

\title{\LARGE \bf
A Predictive Framework for Adversarial Energy Depletion \\ in Inbound Threat Scenarios$^*$\thanks{$^*$Preprint submitted to the American Control Conference (ACC) 2026.}
}

\author{Tam W. Nguyen$^{1}$
\thanks{$^{1}$Tam W. Nguyen is with the Department of Electrical Engineering, Kyoto University, Kyoto 615-8510, Japan
        {\tt\small nguyen.tamwilly.3e@kyoto-u.ac.jp}}%
}

\newtheorem{lemma}{Lemma}

\begin{document}

\maketitle
\thispagestyle{empty}
\pagestyle{empty}

\begin{abstract}

This paper presents a predictive framework for adversarial energy-depletion defense against a maneuverable inbound threat (IT). The IT solves a receding-horizon problem to minimize its own energy while reaching a high-value asset (HVA) and avoiding interceptors and static lethal zones modeled by Gaussian barriers. Expendable interceptors (EIs), coordinated by a central node (CN), maintain proximity to the HVA and patrol centers via radius-based tether costs, deny attack corridors by harassing and containing the IT, and commit to intercept only when a geometric feasibility test is confirmed. No explicit opponent-energy term is used, and the formulation is optimization-implementable. No simulations are included.

\end{abstract}

\section{Introduction}

Adversarial target defense has been studied through pursuit-evasion dynamics and target-attacker-defender (TAD) frameworks. Prior work formulates optimal interception, defender assignment, and reach-avoid strategies under varied performance assumptions \cite{breakwell1975pursuit,eloy2017active,9147205,garcia2021complete}. Model predictive control (MPC) has been applied to one or both agents, typically to minimize miss distance or capture time \cite{1470179,doi:10.2514/1.50572,5756676}. Energy and fuel are treated as resource limits or self-minimizing costs \cite{doi:10.2514/6.2015-0337,7139438,bryson2018applied}. Other approaches, such as herding and containment, redirect threats without terminal intercept \cite{casbeer2018target}.

This paper proposes a predictive framework in which expendable interceptors (EIs) apply sustained pressure to deplete inbound threat (IT) energy through maneuver denial. EIs hold patrol formations and reposition to block attack vectors, engaging only when a geometric feasibility test confirms intercept viability. No explicit opponent-energy term is used; the structure instead exploits maneuver asymmetry and restricted access to induce exhaustion. To the author's knowledge, threat depletion without explicit capture has received limited attention in prior TAD literature.

Both agents solve simplified receding-horizon problems. The IT minimizes its own energy while penetrating toward the high-value asset (HVA). EIs minimize effort while remaining tethered to patrol centers and the HVA via soft positional penalties. When intercepts are viable, they are executed. Otherwise, EIs harass, shadow, and contain. The framework is compact, optimization-compatible, and conceptual. No simulation results are included.

This contribution is theoretical and exploratory. It outlines a control structure for threat exhaustion under asymmetric performance and constrained coverage. The formulation may support future work in autonomous defense, force design, and HVA protection.

\section{Preliminaries}

All $n$-dimensional real vectors $a\in\mathbb{R}^n$ are defined as column vectors.
Denote by $\mathbb{R}_{\ge0}$ the set of nonnegative real numbers, $\mathbb{R}_{>0}$ the set of strictly positive numbers, $\mathbb{N}$ the set of natural numbers without zero, and $\mathbb{N}$ the set of natural numbers including zero.

For conciseness, functions may be written with omitted arguments, i.e., $f(x_1,\ldots,x_n)$ becomes $f(\cdot)$.

All headings are measured counterclockwise from the $x$-axis; the $\mathrm{arctan2}$ function is defined by:
\begin{align}
    \mathrm{arctan2}(y,x) \coloneq 
    \begin{cases}
    \arctan\left(\dfrac{y}{x}\right), & \text{if } x > 0 \\
    \arctan\left(\dfrac{y}{x}\right) + \pi, & \text{if } x < 0,\, y \geq 0 \\
    \arctan\left(\dfrac{y}{x}\right) - \pi, & \text{if } x < 0,\, y < 0 \\
    +\dfrac{\pi}{2}, & \text{if } x = 0,\, y > 0 \\
    -\dfrac{\pi}{2}, & \text{if } x = 0,\, y < 0 \\
    \text{undefined}, & \text{if } x = 0,\, y = 0.
    \end{cases}
\end{align}

Denote by $t\in\mathcal{T}$ the absolute physical time, where $\mathcal{T}\coloneq\mathbb{R}_{\ge0}$ is the continuous physical time domain.

\section{Operational Scenario}

The system consists of three agents:
\begin{itemize}
    \item attacker: highly capable inbound threat (IT);
    \item defender: team of expendable interceptors (EIs);
    \item high-value asset (HVA): target to be defended.
\end{itemize}

The IT is a non-hovering, maneuverable strike asset committed to a one-way attack on the HVA. It seeks to:
\begin{itemize}
    \item reach the HVA;
    \item conserve maneuvering energy;
    \item minimize exposure to static defenses;
    \item evade interception by EIs.
\end{itemize}

The EIs are non-hovering, non-kinetic, expendable agents deployed to deny access to HVA.
Each is tasked to:
\begin{itemize}
    \item neutralize the IT via terminal intercept if feasible;
    \item otherwise, degrade the IT's maneuvering capability through sustained engagement, leading to mission failure.
\end{itemize}

The HVA is a fixed strategic asset located at $p_{\text{hva}} = \begin{bmatrix} x_{\text{hva}} & y_{\text{hva}}\end{bmatrix}^\top \in \mathbb{R}^2$.

\section{Threat Modeling}

For all $t\ge0$, the IT follows the dynamics:
\begin{align}
    \dot{x}(t) & = v(t)\cos(\theta(t)) \label{eq:IT1} \\
    \dot{y}(t) & = v(t)\sin(\theta(t)) \label{eq:IT2} \\
    \dot{v}(t) & = a(t) \label{eq:IT3} \\
    \dot{\theta}(t) & = \omega(t), \label{eq:IT4}
\end{align}
where:
\begin{itemize}
    \item $p(t) \coloneq \begin{bmatrix} x(t) & y(t) \end{bmatrix}^\top \in \mathbb{R}^2$: IT position;
    \item $v(t)\in[v_{\text{min}},v_{\text{max}}] \subset \mathbb{R}$: IT speed;
    \item $\theta(t) \in \mathbb{R}$: IT heading;
    \item $a(t) \in [a_{\text{min}}, a_{\text{max}}] \subset \mathbb{R}$: IT longitudinal acceleration;
    \item $\omega(t) \in [\omega_{\text{min}}, \omega_{\text{max}}] \subset \mathbb{R}$: IT turn rate.
\end{itemize}

The system enforces all constraints as hard limits.
The IT maintains forward motion at all times (no hovering or reversal permitted):
\begin{align}
    v_{\text{min}} \in \mathbb{R}_{>0}.
\end{align}

For all $t\ge0$, the IT's remaining maneuvering energy $e(t)$ is governed by:
\begin{align}
    \dot{e}(t) = -\left(a(t)^2 + \lambda\omega(t)^2\right), \quad e(t) \in \mathbb{R}_{\ge0}, \label{eq:e}
\end{align}
where:
\begin{itemize}
    \item $\dot{e}(t)$: IT instantaneous rate of energy depletion;
    \item $\lambda \in \mathbb{R}_{>0}$: IT penalty of turning relative to forward acceleration.
\end{itemize}

Let $\hat{t}\in\hat{\mathcal{T}}\subset \mathcal{T}$ denote the IT's current engagement time and $\hat{\tau}\in\hat{\mathcal{T}}$ its tactical horizon.
The precise sampling and time-indexing structure is detailed in Section~\ref{sec:it_planner}.

At each time $\hat{t}$, the IT minimizes the composite cost over the engagement window $[\hat{t},\hat{t}+\hat{\tau}]$:
\begin{align}
    \mathcal{J}(\hat{t}) \coloneq m_1 E(\hat{t}) + m_2 R(\hat{t}) + m_3 \bar{D}(\hat{t}),
\end{align}
where:
\begin{itemize}
    \item $E(\hat{t}) \coloneq -\int_{\hat{t}}^{\hat{t}+\hat{\tau}} \dot{e}(t)dt$: IT cumulative maneuvering energy expended over the engagement window;
    \item $R(\hat{t}) \coloneq \int_{\hat{t}}^{\hat{t}+\hat{\tau}} \rho(t)dt$: IT cumulative risk from exposure to static defenses and proximity to EIs;
    \item $\rho(t)\in\mathbb{R}_{>0}$: instantaneous risk density;
    \item $\bar{D}(\hat{t}) \coloneq \|p(\hat{t}+\hat{\tau}) - p_{\text{hva}}\|^2$: IT squared terminal distance to HVA;
    \item $m_j \in \mathbb{R}_{>0}$: IT mission weights, with $j\in\{1,2,3\}$.
\end{itemize}

Implementation details are provided in Section~\ref{sec:it_planner}; $\rho(t)$ is defined in Section~\ref{sec:it_risk}.

\section{Interceptor Modeling}

The defender deploys $n$ EIs under a centralized command node (CN).
For all $t\ge0$, each EI, indexed $i\in\mathcal{I}\coloneq\{1,\ldots,n\}$, follows:
\begin{align}
    \dot{x}^i(t) &= v^i(t) \cos(\theta^i(t)) \label{eq:xidot} \\
    \dot{y}^i(t) &= v^i(t) \sin(\theta^i(t)) \label{eq:yidot} \\
    \dot{v}^i(t) &= a^i(t) \label{eq:vidot} \\
    \dot{\theta}^i(t) &= \omega^i(t),  \label{eq:thetaidot}
\end{align}
where:
\begin{itemize}
    \item $p^i(t) \coloneq \begin{bmatrix} x^i(t) & y^i(t) \end{bmatrix}^\top \in \mathbb{R}^2$: EI $i$ position;
    \item $v^i(t) \in [\underline{v},\overline{v}] \subset \mathbb{R}$: EI $i$ speed;
    \item $\theta^i(t) \in \mathbb{R}$: EI $i$ heading;
    \item $a^i(t) \in [\underline{a}, \overline{a}] \subset \mathbb{R}$: EI $i$ longitudinal acceleration;
    \item $\omega^i(t) \in [\underline{\omega},\overline{\omega}]\subset \mathbb{R}$: EI $i$ turn rate.
\end{itemize}

The system enforces all constraints as hard limits.
Each EI maintains forward motion at all times:
\begin{align}
    \underline{v} \in \mathbb{R}_{>0}.
\end{align}

For all $t\ge0$, the EI $i$'s remaining maneuvering energy $e^i(t)$ is governed by:
\begin{align}
    \dot{e}^i(t) = -\left(\{a^i(t)\}^2 + \kappa\{\omega^i(t)\}^2\right), \quad e^i(t) \in \mathbb{R}_{\ge0}, \label{eq:eidot}
\end{align}
where:
\begin{itemize}
    \item $\dot{e}^i(t)$: EI $i$ instantaneous rate of energy depletion;
    \item $\kappa \in \mathbb{R}_{>0}$: EI penalty of turning relative to forward acceleration.
\end{itemize}

All EIs share a synchronized engagement time $\tilde{t}\in\tilde{\mathcal{T}}\subset\mathcal{T}$ and a tactical horizon $\tilde{\tau}\in\tilde{\mathcal{T}}$, imposed by the CN.
The precise sampling and time-indexing structure is detailed in Section~\ref{sec:ei_planner}.


At each time $\tilde{t}$, the CN identifies a candidate set $\mathcal{N}(\tilde{t})$ of EIs within terminal intercept range.
For each $i\in\mathcal{N}(\tilde{t})$, the CN evaluates terminal intercept feasibility.
If feasible, terminal engagement is executed; otherwise the EI is excluded from $\mathcal{N}(\tilde{t})$ and reassigned to pursuit proximity.
Implementation of the terminal planner is detailed in Section~\ref{sec:terminal_planner}.

Next, the CN minimizes the pursuit proximity cost over the engagement window $[\tilde{t},\tilde{t}+\tilde{\tau}]$:
\begin{align}
    \mathscr{J}(\tilde{t}) \coloneq \sum_{i\in\mathcal{P}(\tilde{t})}\left(\mu_1 E^i(\tilde{t}) + \mu_2 B^i(\tilde{t}) + \mu_3 P^i(\tilde{t})\right),
\end{align}
where:
\begin{itemize}
    \item $E^i(\tilde{t}) \coloneq -\int_{\tilde{t}}^{\tilde{t}+\tilde{\tau}} \dot{e}^i(t) dt$: EI $i$ cumulative maneuvering energy expended over the engagement window;
    \item $B^i(\tilde{t}) \in\mathbb{R}_{>0}$: EI $i$ barrier term penalizing deviation from patrol zone and HVA defense perimeter;
    \item $P^i(\tilde{t})\coloneq\int_{\tilde{t}}^{\tilde{t}+\tilde{\tau}}\{\delta^i(t)\}^2dt$:
    EI $i$ cumulative squared proximity cost;
    \item $\delta^i(t) \in \mathbb{R}_{\ge0}$: instantaneous proximity distance between IT and EI $i$;
    \item $\mu_j \in \mathbb{R}_{>0}$: EI mission weights, with $j\in\{1,2,3\}$;
    \item $\mathcal{P}(\tilde{t}) \coloneq \mathcal{I}\setminus\mathcal{N}(\tilde{t})$: pursuit proximity set; $\mathcal{N}(\tilde{t})$ is updated following the terminal planner execution.
\end{itemize}

$\delta^i(t)$ is defined in Section~\ref{sec:terminal_planner} and the implementation of the pursuit proximity planner is detailed in Section~\ref{sec:pursuit_planner}; $B^i(\tilde{t})$ is defined in Section~\ref{sec:ei_barrier}.

\section{Engagement Outcomes}

Let $r_{\text{dz}} \in \mathbb{R}_{>0}$ denote the IT terminal dive zone (DZ) radius.  
Let $r_{\text{iz}} \in \mathbb{R}_{>0}$ denote the terminal intercept zone (IZ) radius.  
An EI $i$ achieves terminal intercept if $\|p - p^i\| \le r_{\text{iz}}$.

The IT is neutralized if intercepted within the IZ or fully depleted before entering the DZ.

EI terminal detonation is triggered autonomously via continuous-time proximity logic upon IZ entry, independent of CN control cycles.
IT impact is registered via continuous-time DZ boundary monitoring, independent of trajectory sampling.

Table~\ref{tab:outcome} summarizes the engagement outcomes.

\begin{table}[!ht]
\centering
\caption{Engagement Outcomes}\label{tab:outcome}
\begin{tabularx}{\linewidth}{|c|X|X|}
\hline
\textbf{Agent} & \textbf{Success} & \textbf{Failure} \\
\hline
\textbf{IT} & 
$ \|p - p_{\text{hva}}\| \le r_{\text{dz}}$ & 
$ e = 0 $ or
intercepted
\\
\hline
\textbf{EI} & 
IT neutralized before $\|p-p_{\text{hva}}\|\le r_{\text{dz}}$ & 
IT reaches DZ \\
\hline
\textbf{HVA} & 
IT denied access to DZ & 
IT reaches DZ; HVA lost \\
\hline
\end{tabularx}
\end{table}

\section{Predictive Engagement Framework}\label{sec:predict}

\subsection{IT Planner}\label{sec:it_planner}

The IT defines a discrete-time decision grid over the continuous time axis $\mathcal{T}$, with step size $\hat{T}_\mathrm{s}\in\mathbb{R}_{>0},$ yielding decision points at $\hat{t}_k = k \hat{T}_\mathrm{s}$, $k \in \mathbb{N}_0$.
The resulting discrete-time domain is $\hat{\mathcal{T}} \coloneq \{\hat{t}_0, \hat{t}_1,\ldots\}\subset \mathcal{T}$.


Applying forward Euler integration to \eqref{eq:IT1}--\eqref{eq:IT4} and \eqref{eq:e} with step size $\hat{T}_\mathrm{s}$ yields:
\begin{align}
    s_{k+1} &= s_k + \hat{T}_\mathrm{s} f(s_k,u_k) \label{eq:sdis} \\
    e_{k+1} &= e_k - \hat{T}_\mathrm{s} \left(a_k^2 + \lambda \omega_k^2\right), \label{eq:edis}
\end{align}
where:
\begin{itemize}
    \item $s_k \coloneq \begin{bmatrix}p_k^\top & v_k & \theta_k\end{bmatrix}^\top \in \mathbb{R}^4$: IT state at step $k$;
    \item $p_k = \begin{bmatrix}x_k & y_k\end{bmatrix}^\top\in\mathbb{R}^2$: IT position at step $k$;
    \item $u_k \coloneq \begin{bmatrix} a_k & \omega_k \end{bmatrix}^\top\in\mathbb{R}^2$: IT control, held constant over the interval $[\hat{t}_k,\hat{t}_{k+1})$;
    \item $e_k \in \mathbb{R}_{\ge0}$: IT remaining maneuvering energy at step $k$;
    \item $f(\cdot) \coloneq \begin{bmatrix} v_k \cos\theta_k & v_k \sin\theta_k & a_k & \omega_k \end{bmatrix}^\top \in \mathbb{R}^4$.
\end{itemize}

For all $k\ge0$, the IT constraints are:
\begin{align}
    s_k \in \mathcal{S}; \quad u_k \in \mathcal{U}; \quad e_k \in \mathcal{E}, \label{eq:itcon}
\end{align}
where:
\begin{itemize}
    \item $\mathcal{S} \coloneq \mathbb{R}^2\times[v_\text{min},v_\text{max}]\times\mathbb{R}$: IT state constraint set;
    \item $\mathcal{U} \coloneq [a_\text{min},a_\text{max}]\times[\omega_\text{min},\omega_\text{max}]$: IT control constraint set;
    \item $\mathcal{E} \coloneq \mathbb{R}_{\ge0}$: IT energy constraint set.
\end{itemize}

For all $k\ge0$, define the discrete cost:
\begin{align}
    \mathcal{J}_k \coloneq & \sum_{i=0}^{h_{\hat{\tau}}-1} \left(\ell(s_{i|k},u_{i|k}) + \nu \|\sigma_{i|k}\|^2\right)\hat{T}_\mathrm{s} \nonumber \\ & + \phi(s_{h_{\hat{\tau}}|k}) + \nu \|\tilde{\sigma}_{h_{\hat{\tau}}|k}\|^2,
\end{align}
where:
\begin{itemize}
    \item $h_{\hat{\tau}} \in \mathbb{N}$: IT tactical horizon, with $h_{\hat{\tau}}\hat{T}_\mathrm{s} = \hat{\tau}$;
    \item $s_{i|k} \in \mathbb{R}^4$: IT state at step $k+i$, predicted at step $k$;
    \item $u_{i|k} \coloneq \begin{bmatrix}
        a_{i|k} & \omega_{i|k}
    \end{bmatrix}^\top \in \mathbb{R}^2$: IT computed control, held over the interval $[\hat{t}_{k+i},\hat{t}_{k+i+1})$;
    \item $\ell(\cdot) \coloneq m_1 \left(a_{i|k}^2 + \lambda\omega_{i|k}^2\right) + m_2 \rho_{i|k}$: IT stage cost;
    \item $\rho_{i|k} \in \mathbb{R}_{>0}$: IT instantaneous risk density at step $k+i$, predicted at step $k$; defined in \eqref{eq:rho_dis};
    \item $\phi(\cdot) \coloneq m_3 \|p_{h_{\hat{\tau}}|k}-p_{\text{hva}}\|^2$: IT squared terminal distance-to-HVA cost;
    \item $\nu \in \mathbb{R}_{>0}$: IT slack weight;
    \item $\sigma_{i|k} \coloneq \begin{bmatrix}(\sigma^s_{i|k})^\top & (\sigma^u_{i|k})^\top\end{bmatrix}^\top \in \mathbb{R}^6$: IT stage slack at step $k+i$, predicted at step $k$, with $\sigma^s_{i|k}\in\mathbb{R}^4$ and $\sigma^u_{i|k}\in\mathbb{R}^2$;
    \item $\tilde{\sigma}_{h_{\hat{\tau}}|k} \in \mathbb{R}^4$: IT terminal slack predicted at step $k$.
\end{itemize}

At each step $k$, the IT solves the optimization problem:
\begin{align}
    \min_{\{u_{i|k},\sigma_{i|k}\}_{i=0}^{h_{\hat{\tau}}-1},\; \tilde{\sigma}_{h_{\hat{\tau}}|k}} \quad & \mathcal{J}_k \nonumber \\
    \text{subject to} \quad & \text{\eqref{eq:sdis}--\eqref{eq:edis}}; \nonumber \\
    & s_{i|k} \in \mathcal{S} \oplus \sigma^s_{i|k} \nonumber\\
    & u_{i|k} \in \mathcal{U} \oplus \sigma^u_{i|k} \nonumber \\ 
    & \sigma_{i|k} \ge 0, \nonumber \\ & \forall i \in \{0, \ldots, h_{\hat{\tau}}-1\}; \nonumber \\
    & s_{h_{\hat{\tau}}|k} \in \mathcal{S} \oplus \tilde{\sigma}_{h_{\hat{\tau}}|k} \nonumber \\ 
    & e_{h_{\hat{\tau}}|k} \in \mathcal{E} \nonumber \\
    & \tilde{\sigma}_{h_{\hat{\tau}}|k} \ge 0,
\end{align}
where:
\begin{itemize}
    \item $e_{h_{\hat{\tau}}|k} \in \mathbb{R}_{\ge0}$: IT terminal remaining maneuvering energy predicted at step $k$.
\end{itemize}

Inequalities are enforced componentwise.
Slack augments bounds elementwise, with terminal slack applied to state only.
Constraint violations incur penalties and denote planning failures.

The IT applies $u_{0|k}$ over the interval $[\hat{t}_k, \hat{t}_{k+1})$.

\subsection{EI Planner} \label{sec:ei_planner}

The EI defines a discrete-time decision grid over the continuous time axis $\mathcal{T}$, with step size $\tilde{T}_\mathrm{s}\in\mathbb{R}_{>0},$ yielding decision points at $\tilde{t}_k = k \tilde{T}_\mathrm{s}$, $k \in \mathbb{N}_0$.
The resulting discrete-time domain is $\tilde{\mathcal{T}} \coloneq \{\tilde{t}_0, \tilde{t}_1,\ldots\}\subset \mathcal{T}$.


For each $i\in\mathcal{I}$, applying forward Euler integration to \eqref{eq:xidot}--\eqref{eq:thetaidot} and \eqref{eq:eidot} with step size $\tilde{T}_\mathrm{s}$ yields:
\begin{align}
    s^i_{k+1} &= s^i_k + \tilde{T}_\mathrm{s} f^i(s^i_k,u^i_k) \label{eq:sidis} \\
    e^i_{k+1} &= e^i_k - \tilde{T}_\mathrm{s} \left((a^i_k)^2 + \kappa (\omega^i_k)^2\right), \label{eq:eidis}
\end{align}
where:
\begin{itemize}
    \item $s^i_k \coloneq \begin{bmatrix}(p^i_k)^\top & v^i_k & \theta^i_k\end{bmatrix}^\top \in \mathbb{R}^4$: EI $i$ state at step $k$;
    \item $p^i_k = \begin{bmatrix}x^i_k & y^i_k\end{bmatrix}^\top\in\mathbb{R}^2$: EI $i$ position at step $k$;
    \item $u^i_k \coloneq \begin{bmatrix} a^i_k & \omega^i_k \end{bmatrix}^\top \in \mathbb{R}^2$: EI $i$ control, held constant over the interval $[\tilde{t}_k,\tilde{t}_{k+1})$;
    \item $e^i_k \in \mathbb{R}_{\ge0}$: EI $i$ remaining maneuvering energy at step $k$;
    \item $f^i(\cdot) \coloneq \begin{bmatrix} v^i_k \cos\theta^i_k & v^i_k \sin\theta^i_k & a^i_k & \omega^i_k \end{bmatrix}^\top \in \mathbb{R}^4$.
\end{itemize}

For all $k\ge0$ and $i\in\mathcal{I}$, the EI $i$ constraints are:
\begin{align}
    s^i_k \in \mathscr{S}; \quad
    u^i_k \in \mathscr{U}; \quad
    e^i_k \in \mathscr{E}, \label{eq:eicon}
\end{align}
where:
\begin{itemize}
    \item $\mathscr{S} \coloneq \mathbb{R}^2\times[\underline{v},\overline{v}]\times\mathbb{R}$: EI state constraint set;
    \item $\mathscr{U} \coloneq [\underline{a},\overline{a}]\times[\underline{\omega},\overline{\omega}]$: EI control constraint set;
    \item $\mathscr{E} \coloneq \mathbb{R}_{\ge0}$: EI energy constraint set.
\end{itemize}

\subsubsection{Terminal Intercept Planner} \label{sec:terminal_planner}

At each step $k$, the CN identifies a candidate set $\mathcal{N}_k\coloneq \left\{i \in \mathcal{I}: \|p^i_k-p_k\|\le r_{\text{pz}} \right\}$ of EIs within the IT's proximity zone of radius $r_{\text{pz}}$.
If $\mathcal{N}_k = \emptyset$, the CN aborts terminal intercept planning and reverts to the pursuit proximity planner.

The CN models the IT as fully committed to terminal attack against the HVA.
At each sampling time $\tilde{t}_k$, the CN simulates the IT's future trajectory for all $t \ge \tilde{t}_k$, assuming nominal attack velocity and a constant heading $\theta^{\text{trm}}_k$ computed at $\tilde{t}_k$.
The dynamics are:
\begin{align}
    \dot{\tilde{p}}(t) &= v_{\text{atk}} \begin{bmatrix}
        \cos \theta^{\text{trm}}_k \\
        \sin \theta^{\text{trm}}_k
    \end{bmatrix}, \label{eq:it_anticipate}
\end{align}
where:
\begin{itemize}
    \item $\tilde{p}(t) \coloneq \begin{bmatrix}
        \tilde{x}(t) & \tilde{y}(t)
    \end{bmatrix}^\top\in\mathbb{R}^2$: IT position anticipated at time $t\ge \tilde{t}_k$;
    \item $\theta^{\text{trm}}_k \coloneq (1-\psi)\theta^{\text{atk}}_k + \psi \theta^{\text{evd}}_k$: IT terminal heading computed at time $\tilde{t}_k$;
    \item $\theta^{\text{atk}}_k \coloneq \mathrm{arctan2}(y_{\text{hva}}-y_k,x_{\text{hva}}-x_k)$: IT attack heading toward HVA at time $\tilde{t}_k$;
    \item $\theta^{\text{evd}}_k \coloneq \frac{1}{n_{\text{pz}}} \sum\limits_{i\in\mathcal{N}_k} \theta^{i,\text{evd}}_k$: IT average evasion heading at time $\tilde{t}_k$; $n_{\text{pz}}$ denotes the cardinality of $\mathcal{N}_k$;
    \item $\theta^{i,\text{evd}}_k \coloneq \mathrm{arctan2}(y^i_k-y_k, x^i_k -  x_k)$: IT evasion heading to evade EI $i$ at time $\tilde{t}_k$;
    \item $\psi \in [0,1]$: Weight assigned to attack relative to evasion;
    \item $v_{\text{atk}}\in\mathbb{R}_{>0}$: IT terminal attack velocity.
\end{itemize}

Applying forward Euler integration to \eqref{eq:it_anticipate} with step size $\tilde{T}_\mathrm{s}$ yields:
\begin{align}
    \tilde{p}_{j+1|k} & = \tilde{p}_{j|k} + \tilde{T}_\mathrm{s} v_{\text{atk}} \begin{bmatrix}
        \cos \theta^{\text{trm}}_k \\
        \sin \theta^{\text{trm}}_k
    \end{bmatrix}, \label{eq:ITattack}
\end{align}
where:
\begin{itemize}
    \item $j \in \mathbb{N}_0$: predictive step;
    \item $\tilde{p}_{j|k} \coloneq \begin{bmatrix}
        \tilde{x}_{j|k} & \tilde{y}_{j|k}
    \end{bmatrix}^\top\in\mathbb{R}^2$: IT position at step $k+j$, anticipated at step $k$.
\end{itemize}
%

For all $k\ge0$ and $i\in\mathcal{N}_k$, define the discrete cost:
\begin{align}
    \mathscr{J}^i_k \coloneq &  \sum_{j=0}^{h_{\tilde{\tau}}-1} (\delta^i_{j|k})^2\tilde{T}_\mathrm{s},
\end{align}
where:
\begin{itemize}
    \item $h_{\tilde{\tau}} \in \mathbb{N}$: EI tactical horizon, with $h_{\tilde{\tau}} \tilde{T}_\mathrm{s} = \tilde{\tau}$;
    \item $\delta^i_{j|k}\coloneq\|p^i_{j|k}-\tilde{p}_{j|k}\|$: instantaneous separation distance at step $k+j$, predicted at step $k$;
    \item $p^i_{j|k} \coloneq \begin{bmatrix}
        x^i_{j|k} & y^i_{j|k}
    \end{bmatrix}^\top\in\mathbb{R}^2$: EI $i$ position at step $k+j$, predicted at step $k$.
\end{itemize}

At step $k$, for each $i\in\mathcal{N}_k$, the CN solves the optimization problem:
\begin{align}
    \min_{\{u^i_{j|k}\}_{j=0}^{h_{\tilde{\tau}}-1}} \quad &\mathscr{J}^i_k \nonumber \\
    \text{subject to} \quad & \text{\eqref{eq:sidis}--\eqref{eq:eidis}, \eqref{eq:ITattack}}; \nonumber \\
    & s^i_{j|k} \in \mathscr{S} \nonumber\\
    & u^i_{j|k} \in \mathscr{U}, \nonumber \\
    & \forall j \in \{0, \ldots, h_{\tilde{\tau}}-1\}; \nonumber \\
    & s^i_{h_{\tilde{\tau}}|k} \in \mathscr{S}^{\text{trm}}_k \nonumber \\ 
    & e^i_{h_{\tilde{\tau}}|k} \in \mathscr{E}, \label{eq:terminal_intercept}
\end{align}
where:
\begin{itemize}
    \item $s^i_{j|k} \in \mathbb{R}^4$: EI $i$ state at step $k+j$, predicted at step $k$;
    \item $u^i_{j|k} \coloneq \begin{bmatrix}
        a^i_{j|k} & \omega^i_{j|k}
    \end{bmatrix}^\top \in \mathbb{R}^2$: EI $i$ computed control, held over the interval $[\tilde{t}_{k+j},\tilde{t}_{k+j+1})$;
    \item $e^i_{h_{\tilde{\tau}}|k} \in \mathbb{R}_{\ge0}$: EI $i$ terminal remaining maneuvering energy predicted at step $k$;
    \item $\mathscr{S}^{\text{trm}}_k \coloneq \{\zeta\in\mathbb{R}^2:\|\zeta-\tilde{p}_{h_{\tilde{\tau}}|k}\|\le r_{\text{iz}}\} \times [\underline{v},\overline{v}] \times \mathbb{R}$: EI state terminal constraint set computed at step $k$.
\end{itemize}

If \eqref{eq:terminal_intercept} is infeasible, EI $i$ is removed from $\mathcal{N}_k$ and assigned to the pursuit proximity set.

\subsubsection{Pursuit Proximity Planner}\label{sec:pursuit_planner}

Let $\mathcal{P}_k\coloneq\mathcal{I}\setminus\mathcal{N}_k$ denote the pursuit proximity set, updated after the terminal planner execution.

For all $k\ge0$, define the discrete cost:
\begin{align}
    \mathscr{J}_k \coloneq & \sum_{i\in\mathcal{P}_k}\bigg\{\sum_{j=0}^{h_{\tilde{\tau}}-1}\left(\ell^i(s^i_{j|k}, u^i_{j|k}) + \eta \|\sigma^i_{j|k}\|^2\right)\tilde{T}_\mathrm{s} \nonumber \\
    & + \eta \sum_{i=1}^n \|\tilde{\sigma}^i_{h_{\tilde{\tau}}|k}\|^2\bigg\},
\end{align}
where:
\begin{itemize}
    \item $\ell^i(\cdot) \coloneq \mu_1 \left(\{a^i_{j|k}\}^2 + \kappa\{\omega^i_{j|k}\}^2\right) + \mu_2 B^i_{j|k} + \mu_3 \{\delta^i_{j|k}\}^2$: EI $i$ stage cost;
    \item $B^i_{j|k} \in \mathbb{R}_{\ge0}$: 
    EI $i$ barrier term at step $k+j$, predicted at step $k$; defined in \eqref{eq:barrier};
    \item $\eta \in \mathbb{R}_{>0}$: EI slack weight;
    \item $\sigma^i_{j|k} \coloneq \begin{bmatrix}(\sigma^{s,i}_{j|k})^\top & (\sigma^{u,i}_{j|k})^\top\end{bmatrix}^\top \in \mathbb{R}^6$: EI $i$ stage slack at step $k+j$, predicted at step $k$, with $\sigma^{s,i}_{j|k}\in\mathbb{R}^4$ and $\sigma^{u,i}_{j|k}\in\mathbb{R}^2$;
    \item $\tilde{\sigma}^i_{h_{\tilde{\tau}}|k} \in \mathbb{R}^4$: EI $i$ terminal slack predicted at step $k$.
\end{itemize}

At step $k$, the CN solves the optimization problem:
\begin{align}
    \min_{\left\{\{u^i_{j|k},\sigma^i_{j|k}\}_{j=0}^{h_{\tilde{\tau}}-1},\tilde{\sigma}^i_{h_{\tilde{\tau}}|k}\right\}_{i\in\mathcal{P}_k}} \quad & \mathscr{J}_k \nonumber \\
    \text{subject to} \quad & \text{\eqref{eq:sidis}--\eqref{eq:eidis}, \eqref{eq:ITattack}}; \nonumber \\
    & s^i_{j|k} \in \mathscr{S} \oplus \sigma^{s,i}_{j|k} \nonumber\\
    & u^i_{j|k} \in \mathscr{U} \oplus \sigma^{u,i}_{j|k} \nonumber \\ 
    & \sigma^i_{j|k} \ge 0, \nonumber \\ & \forall j \in \{0, \ldots, h_{\tilde{\tau}}-1\}; \nonumber \\
    & s^i_{h_{\tilde{\tau}}|k} \in \mathscr{S} \oplus \tilde{\sigma}^i_{h_{\tilde{\tau}}|k} \nonumber \\ 
    & e^i_{h_{\tilde{\tau}}|k} \in \mathscr{E} \nonumber \\
    & \tilde{\sigma}^i_{h_{\tilde{\tau}}|k} \ge 0. \label{eq:proximity_pursuit}
\end{align}

Constraint violations incur penalties and denote planning failures.

\subsubsection{CN Engagement Decision and Execution}

At time $\tilde{t}_k$, for each $i \in \mathcal{N}_k$, the CN solves \eqref{eq:terminal_intercept}. If feasible, EI $i$ commits to terminal engagement.
EIs that fail terminal feasibility are removed from $\mathcal{N}_k$ and assigned to $\mathcal{P}_k$.

Next, the CN solves \eqref{eq:proximity_pursuit} over the remaining EIs.
For each $i \in \mathcal{I}$, the CN transmits $u^i_{0|k}$, executed over $[\tilde{t}_k, \tilde{t}_{k+1})$ by EI $i$.

Fig.~\ref{fig:cn_decision_flow} illustrates the CN engagement decision flow.

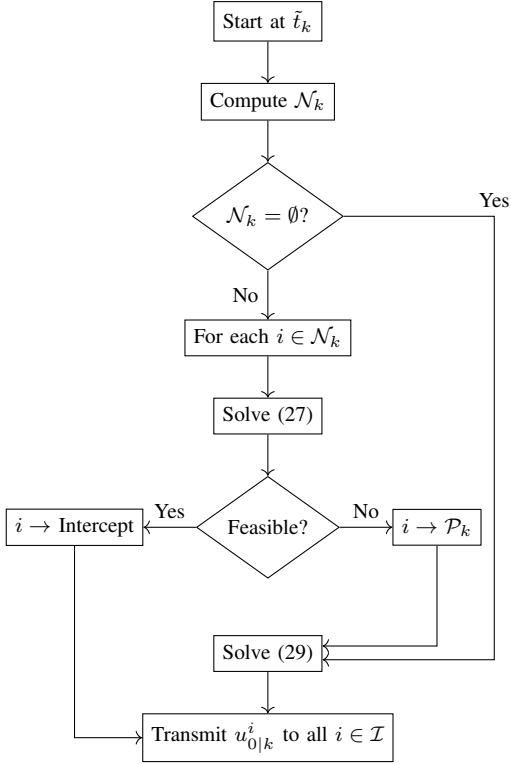
\begin{figure}[H]
    \centering
    \begin{tikzpicture}[
        node distance=0.55cm and 0.4cm,
        every node/.style={font=\footnotesize},
        decision/.style={diamond, aspect=1.4, draw, align=center},
        block/.style={rectangle, draw, align=center}
    ]
        \node (start) [block] {Start at \( \tilde{t}_k \)};
        \node (candidate) [block, below=of start] {Compute \( \mathcal{N}_k \)};
        \node (nullcheck) [decision, below=of candidate] {\( \mathcal{N}_k = \emptyset \)?};

        \node (forLoop) [block, below=of nullcheck, yshift=-0.1cm] {For each \( i \in \mathcal{N}_k \)};
        \node (solveTerm) [block, below=of forLoop] {Solve \eqref{eq:terminal_intercept}};
        \node (termFeas) [decision, below=of solveTerm] {Feasible?};

        \node (commit) [block, left=of termFeas, xshift=-0.3cm] {$i\to$ Intercept};
        \node (exclude) [block, right=of termFeas, xshift=0.3cm] {\( i \to \mathcal{P}_k \)};

        \node (pursuit) [block, below=of termFeas, yshift=-0.2cm] {Solve \eqref{eq:proximity_pursuit}};
        \node (sendU) [block, below=of pursuit] {Transmit \( u^i_{0|k} \) to all \( i \in \mathcal{I} \)};

        \draw[->] (start) -- (candidate);
        \draw[->] (candidate) -- (nullcheck);

        \draw[->] 
        (nullcheck.east) 
        -- ++(2.0,0) node[anchor=south] {Yes}
        -- ++(0,-5.9) 
        coordinate (aux)
        -- ++(-1.6,0) 
        -- (pursuit.east |- aux);

        \draw[->] (nullcheck) -- node[anchor=east] {No} (forLoop);
        \draw[->] (forLoop) -- (solveTerm);
        \draw[->] (solveTerm) -- (termFeas);
        \draw[->] (termFeas) -- node[anchor=south] {Yes} (commit);
        \draw[->] (termFeas) -- node[anchor=south] {No} (exclude);

        \draw[->] (commit.south) |- (sendU.west);

        \draw[->] (exclude.south) |- ([yshift=0.1cm] pursuit.east);
        \draw[->] (pursuit) -- (sendU);
    \end{tikzpicture}
    \caption{CN engagement decision flow at decision time \( \tilde{t}_k \).}
    \label{fig:cn_decision_flow}
\end{figure}

\section{IT Risk and EI Barrier Modeling} \label{sec:risk_barrier}

\subsection{IT Instantaneous Risk Density}\label{sec:it_risk}

The defender deploys $n_{\text{sd}}$ static defenses located at $p^\ell_{\text{sd}} \coloneq \begin{bmatrix} x^\ell_{\text{sd}} & y^\ell_{\text{sd}} \end {bmatrix}^\top \in\mathbb{R}^2$, with $\ell\in\mathcal{D}\coloneq\{1,\ldots,n_{\text{sd}}\}$.
The attacker possesses full intelligence on all static asset locations, acquired via AWACS.

Within the IT's proximity zone of radius $r_{\text{pz}}$, the IT models each EI as fully committed to terminal intercept along the line connecting the EI's current position to the IT's anticipated hitpoint under nominal attack motion.

The IT is modeled as fully committed to terminal attack against the HVA.
At each sampling time $\hat{t}_k$, it simulates its own trajectory for all $t\ge\hat{t}_k$, assuming constant attack velocity and a fixed heading attack $\theta^{\text{atk}}_k$ computed at $\hat{t}_k$. The dynamics are:
\begin{align}
    \dot{\hat{p}}(t) = v_{\text{atk}} \begin{bmatrix}
        \cos\theta^{\text{atk}}_k \\
        \sin\theta^{\text{atk}}_k
    \end{bmatrix}, \label{eq:IT_attack_predict}
\end{align}
where:
\begin{itemize}
    \item $\hat{p}(t)\coloneq \begin{bmatrix}
        \hat{x}(t) & \hat{y}(t)
    \end{bmatrix}^\top\in\mathbb{R}^2$: IT position anticipated at time $t\ge\hat{t}_k$.
\end{itemize}

In parallel, the IT models each proximal EI as incoming, committed to a fixed-trajectory terminal intercept.
At each sampling time $\hat{t}_k$, for each $i\in\mathcal{N}_k$, the IT simulates the trajectory of EI $i$ for all $t\ge\hat{t}_k$, assuming straight-line intercept motion at terminal velocity $v_{\text{itc}}$ with a fixed heading $\theta^{i,\text{itc}}_k$ computed at $\hat{t}_k$. The dynamics are:
\begin{align}
    \dot{\hat{p}}^i(t) = v_{\text{itc}} \begin{bmatrix}
        \cos\theta^{i,\text{itc}}_k \\
        \sin\theta^{i,\text{itc}}_k
    \end{bmatrix}, \label{eq:EI_intercept_predict}
\end{align}
where:
\begin{itemize}
    \item $v_{\text{itc}} \in \mathbb{R}_{>0}$: EI terminal intercept velocity;
    \item $\hat{p}^i(t) \coloneq \begin{bmatrix}\hat{x}^i(t) & \hat{y}^i(t)\end{bmatrix}^\top \in \mathbb{R}^2$: EI $i$ position anticipated at time $t\ge\hat{t}_k$;
    \item $\theta^{i,\text{itc}}_k\in\mathbb{R}$: EI $i$ heading toward the anticipated IT hitpoint, computed at time $\hat{t}_k$.
\end{itemize}

\begin{lemma}[Intercept Heading Feasibility]
At time $\hat{t}_k$, consider the IT position $\hat{p}(\hat{t}_k)$ and a proximal EI $i \in \mathcal{N}_k$ located at $\hat{p}^i(\hat{t}_k)$.  
Assume the IT proceeds at constant speed $v_{\text{atk}}$ along heading $\theta^{\text{atk}}_k$, and the EI engages at constant speed $v_{\text{itc}}$.

Define the line-of-sight (LOS) angle:
\begin{align}
    \theta^{\text{los}}_k \coloneq \arctan2\left( \hat{y}(\hat{t}_k) - \hat{y}^i(\hat{t}_k), \, \hat{x}(\hat{t}_k) - \hat{x}^i(\hat{t}_k) \right),
\end{align}
and the intercept feasibility scalar:
\begin{align}
    \gamma_k \coloneq \frac{v_{\text{atk}}}{v_{\text{itc}}} \sin\left( \theta^{\text{atk}}_k - \theta^{\text{los}}_k \right). \label{eq:gamma_k}
\end{align}

Then:
\begin{itemize}
    \item If $|\gamma_k| > 1 \Rightarrow$ terminal collision is kinematically infeasible; EI $i$ is discarded from $\mathcal{N}_k$.
    \item If $|\gamma_k| \le 1 \Rightarrow$ a unique terminal intercept heading is defined:
    \begin{align}
        \theta^{i,\text{itc}}_k = \theta^{\text{los}}_k + \arcsin(\gamma_k),
    \end{align}
    with corresponding time-to-intercept:
    \begin{align}
        t^{\text{itc}}_k = \frac{ \left\| \hat{p}(\hat{t}_k) - \hat{p}^i(\hat{t}_k) \right\| }{
        \left\| 
            v_{\text{itc}} \begin{bmatrix}
                \cos(\theta^{i,\text{itc}}_k) \\
                \sin(\theta^{i,\text{itc}}_k)
            \end{bmatrix}
            - 
            v_{\text{atk}} \begin{bmatrix}
                \cos(\theta^{\text{atk}}_k) \\
                \sin(\theta^{\text{atk}}_k)
            \end{bmatrix}
        \right\| }.
    \end{align}
\end{itemize}
\end{lemma}

\begin{proof}
Let $r_k \coloneq \hat{p}(\hat{t}_k) - \hat{p}^i(\hat{t}_k) \in \mathbb{R}^2$ denote the LOS vector.
Let $u^{\text{atk}}_k \coloneq \begin{bmatrix}
    \cos \theta^{\text{atk}}_k \\
    \sin \theta^{\text{atk}}_k
\end{bmatrix}\in \mathbb{R}^2$ and $u^{\text{itc}}_k \coloneq \begin{bmatrix}
    \cos \theta^{i,\text{itc}}_k \\
    \sin \theta^{i,\text{itc}}_k
\end{bmatrix} \in \mathbb{R}^2$ denote unit vectors aligned with headings $\theta^{\text{atk}}_k$ and $\theta^{i,\text{itc}}_k$, respectively.

\begin{enumerate}
\item \textbf{Collision Geometry}:
Terminal intercept requires positional coincidence at some future time $\hat{t}_k + t^{\text{itc}}_k$.  
Substituting constant-velocity dynamics into the position equality:
\[
\hat{p}(\hat{t}_k + t^{\text{itc}}_k) = \hat{p}^i(\hat{t}_k + t^{\text{itc}}_k),
\]
yields:
\[
r_k + t^{\text{itc}}_k v^{\text{rel}}_k = 0,
\]
where:
\begin{align}
v^{\text{rel}}_k \coloneq v_{\text{itc}} u^{\text{itc}}_k - v_{\text{atk}} u^{\text{atk}}_k. \label{eq:vrel}
\end{align}
This forces collinearity between $r_k$ and $v^{\text{rel}}_k$.

\item \textbf{Feasibility Condition}:
Let $u^\perp_k \coloneq \begin{bmatrix} -\sin(\theta^{\text{los}}_k) \\ \cos(\theta^{\text{los}}_k) \end{bmatrix}$ be the unit vector orthogonal to the LOS.
Project $v^{\text{rel}}_k$ onto $u^\perp_k$ and enforce zero component:
\[
v^{\text{rel}}_k \cdot u^\perp_k = 0.
\]
Expand using \eqref{eq:vrel}:
\[
v_{\text{itc}} \sin\left( \theta^{i,\text{itc}}_k - \theta^{\text{los}}_k \right)
= v_{\text{atk}} \sin\left( \theta^{\text{atk}}_k - \theta^{\text{los}}_k \right).
\]
Using \eqref{eq:gamma_k}, collision feasibility requires $|\gamma_k| \le 1$.

\item \textbf{Solution Construction}:
If feasible, the intercept heading is uniquely defined as:
\[
\theta^{i,\text{itc}}_k = \theta^{\text{los}}_k + \arcsin(\gamma_k).
\]
The time-to-intercept is defined by the norm ratio:
\[
t^{\text{itc}}_k = \frac{ \|r_k\| }{ \|v^{\text{rel}}_k\| }.
\]
\end{enumerate}
\end{proof}

Define the IT instantaneous risk density:
\begin{align}
    \rho(t) \coloneq \rho_{\text{sd}}(p(t)) + \rho_{\text{ei}}(p(t),\hat{p}^i(t)), \label{eq:rho_def}
\end{align}
where:
\begin{itemize}
    \item $\rho_{\text{sd}}(\zeta) \coloneq w_{\text{sd}} \sum\limits_{\ell\in\mathcal{D}} \exp\left(-\dfrac{\|\zeta-p^\ell_{\text{sd}}\|^2}{2 \sigma^2_{\text{sd}}}\right)$: cost density from static defense exposure;
    \item $w_{\text{sd}} \in \mathbb{R}_{>0}$: static defense lethality scaling factor;
    \item $\sigma_{\text{sd}} \in \mathbb{R}$: spatial spread of static defense field;
    \item $\rho_{\text{ei}}(\zeta, \zeta^i) \coloneq w_{\text{ei}} \sum\limits_{i\in\mathcal{N}_k} \exp\left(-\dfrac{\|\zeta-\zeta^i\|^2}{2 \sigma^2_{\text{ei}}}\right)$: cost density from intercept-capable EI proximity; $\mathcal{N}_k$ pruned post heading feasibility check;
    \item $w_{\text{ei}} \in \mathbb{R}_{>0}$: scaling factor for EI threat intensity;
    \item $\sigma_{\text{ei}} \in \mathbb{R}$: spatial softness of EI threat field.
\end{itemize}

Applying forward Euler integration to \eqref{eq:EI_intercept_predict} with step size $\hat{T}_\mathrm{s}$ yields:
\begin{align}
    \hat{p}^i_{j+1|k} &= \hat{p}^i_{j|k} + \hat{T}_\mathrm{s} v_{\text{atk}}\begin{bmatrix}
        \cos\theta^{i,\text{itc}}_k \\
        \sin\theta^{i,\text{itc}}_k
    \end{bmatrix},
\end{align}
where:
\begin{itemize}
    \item $j \in \mathbb{N}_0$: predictive step;
    \item $\hat{p}^i_{j|k} \coloneq \begin{bmatrix}
        \hat{x}^i_{j|k} & \hat{y}^i_{j|k}
    \end{bmatrix}^\top\in\mathbb{R}^2$: EI $i$ position at step $k+j$, anticipated at step $k$.
\end{itemize}

Define the discrete instantaneous risk density:
\begin{align}
    \rho_{j|k} \coloneq & w_{\text{sd}}\sum_{\ell\in\mathcal{D}} \exp\left(-\frac{\|p_{j|k} - p_{\text{sd}}^\ell\|^2}{2\sigma_{\text{sd}}^2}\right) \nonumber \\
    & + w_{\text{ei}} \sum_{i\in\mathcal{N}_k}\exp\left(-\frac{\|p_{j|k}-\hat{p}^i_{j|k}\|^2}{2\sigma_{\text{ei}}^2}\right). \label{eq:rho_dis}
\end{align}

\subsection{EI Barrier}\label{sec:ei_barrier}

Each pursuing EI $i\in\mathcal{P}_k$ operates under two barrier costs:
\begin{itemize}
    \item Patrol Adherence Cost (PAC): enforces spatial proximity to its assigned patrol center $p^i_\text{pac} \in \mathbb{R}^2$;
    \item HVA Tethering Cost (HTC): imposes a soft constraint on all EIs to maintain proximity to the HVA.
\end{itemize}
This dual-layered structure prevents EIs from being tactically displaced or lured by IT, preserving coverage near HVA while enabling distributed patrol operations.

Define PAC at prediction index $j$:
\begin{align}
    B^{i,\text{pac}}_{j|k} \coloneq 
    \begin{cases}
        0, & \text{if $\delta^{i,\text{pac}}_{j|k} < r_{\text{pac}}$} \\
        w_{\text{pac}} (\delta^{i,\text{pac}}_{j|k} - r_{\text{pac}})^2, & \text{otherwise},
    \end{cases}
\end{align}
where:
\begin{itemize}
    \item $w_{\text{pac}}\in\mathbb{R}_{>0}$: PAC weight;
    \item $r_{\text{pac}} \in \mathbb{R}_{>0}$: patrol zone radius;
    \item $\delta^{i,\text{pac}}_{j|k} \coloneq \|p^i_{j|k} - p^i_\text{pac}\|$: EI $i$ distance to assigned patrol center.
\end{itemize}

Define HTC at prediction index $j$:
\begin{align}
    B^{i,\text{htc}}_{j|k} \coloneq 
    \begin{cases}
        0, & \text{if $\delta^{i,\text{htc}}_{j|k} < r_{\text{htc}}$} \\
        w_{\text{htc}} (\delta^{i,\text{htc}}_{j|k} - r_{\text{htc}})^2, & \text{otherwise},
    \end{cases}
\end{align}
where:
\begin{itemize}
    \item $w_{\text{htc}}\in\mathbb{R}_{>0}$: HTC weight;
    \item $r_{\text{htc}} \in \mathbb{R}_{>0}$: HVA tether radius;
    \item $\delta^{i,\text{htc}}_{j|k} \coloneq \|p^i_{j|k} - p_\text{hva}\|$: EI $i$ distance to HVA.
\end{itemize}

The total barrier cost is computed by:
\begin{align}
    B^i_{j|k} = B^{i,\text{pac}}_{j|k} + B^{i,\text{htc}}_{j|k}. \label{eq:barrier}
\end{align}

\section{Conclusion}

This paper presented a predictive framework for adversarial energy-depletion defense against a maneuverable inbound threat. Both attacker and defender solve receding-horizon problems; the attacker minimizes energy while evading static defenses and interceptors, while defenders coordinate through soft constraints and conditional intercept logic. Unlike capture-focused strategies, this framework imposes exhaustion without explicit opponent-energy terms. Engagement logic is centralized, cost-continuous, and implementable using standard optimization tools. Future work will evaluate effectiveness through Monte Carlo simulations across force structures and engagement scenarios.


\bibliographystyle{IEEEtran.bst}
\bibliography{bibliography.bib}

\end{document}